%% file: Main.tex
\documentclass[11pt]{article}
\usepackage{graphicx}
\usepackage{latexsym}
\usepackage{amssymb}
\usepackage{amsmath}
\usepackage{amsthm}
\usepackage{forloop}
\usepackage[paper=letterpaper,margin=1in]{geometry}
\usepackage[ruled,vlined,linesnumbered]{algorithm2e}
\usepackage{nicefrac}
\usepackage{paralist}
\usepackage{mleftright}

\newtheorem{theorem}{Theorem}
\newtheorem{lemma}{Lemma}[section]

\newtheorem{corollary}[lemma]{Corollary}

\newtheorem{proposition}[lemma]{Proposition}
\newtheorem{observation}[lemma]{Observation}

\makeatletter
\newtheorem*{rep@theorem}{\rep@title}
\newcommand{\newreptheorem}[2]{%
\newenvironment{rep#1}[1]{%
 \def\rep@title{#2 \ref{##1}}%
 \begin{rep@theorem}}%
 {\end{rep@theorem}}}
\makeatother

\newreptheorem{lemma}{Lemma}

\newcommand{\defcal}[1]{\expandafter\newcommand\csname c#1\endcsname{{\mathcal{#1}}}}
\newcommand{\defbb}[1]{\expandafter\newcommand\csname b#1\endcsname{{\mathbb{#1}}}}
\newcounter{calBbCounter}
\forLoop{1}{26}{calBbCounter}{
    \edef\letter{\Alph{calBbCounter}}
    \expandafter\defcal\letter
		\expandafter\defbb\letter
}

\newcommand{\eps}{\varepsilon}
\newcommand{\ie}{{\it i.e.}}

\newcommand{\nnR}{{\bR_{\geq 0}}}

\newcommand{\cupdot}{\mathbin{\mathaccent\cdot\cup}}

\newcommand{\partT}[2]{{#1}_{#2}}
\newcommand{\RPGreedy}{{\textsf{RPGreedy}}}
\newcommand{\RRGreedy}{{\textsf{RRGreedy}}}
\newcommand{\Split}{{\texttt{Split}}}

\title{Deterministic $(\nicefrac{1}{2}+\eps)$-Approximation for Submodular Maximization over a Matroid}
\author{
Niv Buchbinder\thanks{Dept. of Statistics and Operations Research, Tel Aviv University, Israel. E-mail: niv.buchbinder@gmail.com}
\and
Moran Feldman\thanks{Department of Mathematics and Computer Science, The Open University of Israel. E-mail: moranfe@openu.ac.il}
\and
Mohit Garg\thanks{Department of Mathematics and Computer Science, The Open University of Israel. E-mail: mohitga@openu.ac.il}
}

\begin{document}

\maketitle
\input{Abstract}
\pagenumbering{Alph}
\thispagestyle{empty}
\clearpage
\pagenumbering{arabic}

\input{Introduction}

\input{Preliminaries}
\input{basic-algorithms}
\input{sub-monotone-matroid}

\input{DeterministicGreedy}

\bibliographystyle{plain}
\bibliography{SubmodularMax}

\appendix
\input{MatroidTheorem}

\end{document}

%% file: Abstract.tex
\begin{abstract}
We study the problem of maximizing a monotone submodular function subject to a matroid constraint
and present a {\bf deterministic} algorithm that achieves $(\nicefrac{1}{2}+\eps)$-approximation for the problem. This algorithm is the first deterministic algorithm known to improve over the $\nicefrac{1}{2}$-approximation ratio of the classical greedy algorithm proved by Nemhauser, Wolsely and Fisher in $1978$.

\end{abstract} 

%% file: Introduction.tex
\section{Introduction} \label{sec:introduction}

We study the problem of maximizing a monotone submodular function subject to a matroid constraint in the standard value oracle model.
In this problem we are given a matroid $\cM=(\cN,\cI)$ and oracle access to a non-negative monotone submodular function $f:2^{\cN}\to\nnR$. The goal is to find an independent set $I\in\cI$ that maximizes $f$ (see Section~\ref {sec:preliminaries} for definitions). This problem generalizes several extensively studied problems such as submodular welfare maximization, generalized assignment, Max-$k$-Cover and maximum weight independent set in matroid intersection \cite{CCPV11}.

In a classical paper from 1978, Fisher, Nemhauser and Wolsey \cite{FNW78} showed that a natural greedy algorithm achieves a $\nicefrac{1}{2}$-approximation for the problem.
Nemhauser and Wolsely \cite{NW78} also showed that no polynomial-time algorithm can yield an approximation ratio better than $1-\nicefrac 1 e$ (in the value oracle model).
Three decades later, in a breakthrough result, C\u{a}linescu, Chekuri, P\'al and Vondr\'ak presented an optimal randomized $(1-\nicefrac 1 e)$-approximation algorithm~\cite{CCPV11}.
Their algorithm is intrinsically randomized, and they explicitly asked ``whether a $(1-\nicefrac{1}{e})$-approximation can be obtained using a deterministic algorithm".
Such a result is known for a few special cases of the problem (e.g., Max-$k$-Coverage).
However, prior to this work, for the general problem no deterministic algorithm with an approximation ratio strictly better than $\nicefrac{1}{2}$ was known.

\subsection{Our Result}
In this work we present the first {\bf deterministic} algorithm achieving a $(\nicefrac{1}{2}+\eps)$-approximation for the problem of maximizing a monotone submodular function subject to a matroid constraint. Specifically, we prove the following theorem.

\begin{theorem}\label{thm:main}
There exists a deterministic polynomial time algorithm achieving $0.5008$-approx\-imation for the problem of maximizing a non-negative monotone submodular function subject to a matroid constraint.
\end{theorem}

It is worth mentioning that our algorithm is quite efficient as it makes $O(nk^2)$ queries to the value and independence oracles, where $n$ is the size of the ground set and $k$ is the rank of the matroid. Moreover, excluding the time required for these oracles queries, the time complexity of the algorithm is $\tilde{O}(nk^2 + k T)$, where $T$ is the time required for computing a maximum weight perfect matching in a bipartite graph with $2k$ vertices.

\subsection{Our Technique}

Before discussing the new ideas we developed, we first briefly explain why existing techniques fail to produce our result. A natural approach for obtaining deterministic algorithms is to apply derandomization methods to known randomized algorithms.
While this approach has been extremely successful in the field of approximation algorithms, it encounters two obstacles when applied to submodular maximization problems.
First, the (black box) value oracle makes it difficult to apply known derandomization techniques such as the conditional expectations method.
Second, many of the known algorithms in this area (including the above mentioned randomized $(1 - \nicefrac{1}{e})$-approximation algorithm of C\u{a}linescu et al.~\cite{CCPV11}) are based on optimizing a relaxation whose objective function is an extension of $f$ known as the multilinear extension. Unfortunately, the only known way to evaluate this powerful extension is by randomly sampling the function $f$, which makes any algorithm based on it intrinsically randomized.

In an attempt to avoid the above mentioned relaxation,  Filmus and Ward~\cite{FW14} designed a non-oblivious \textit{local search} algorithm achieving a randomized $(1-\nicefrac 1 e)$-approximation for the problem we consider. As their algorithm is based on local search, it only needs to consider integral solutions. However, the auxiliary potential function that it optimizes is, again, estimated by randomly sampling $f$.

Motivated by the difficulty to derandomize the above algorithms, we base our algorithm on two very simple components that are either already deterministic or are simple (and combinatorial) enough for derandomization. The first of these components is a simple deterministic greedy-like {\em split} algorithm. Given a matroid $\cM$, the split algorithm returns two disjoint sets $A$ and $B$ whose union is a base of $\cM$ and obey two additional properties. The more interesting of which is that the weighted average $\beta \cdot f(A) + (1-\beta)\cdot f(B)$ is large for any given value of $\beta \in [\nicefrac{1}{5},\nicefrac{4}{5}]$---the value is determined by a parameter of the algorithm. As this is a very natural property, we believe our split algorithm might be of independent interest.
The second component we use is the (randomized) {\em Residual Random Greedy} (\RRGreedy) algorithm that was originally described in~\cite{BFNS14}. This is, again, a very simple greedy-like algorithm for submodular maximization subject to a matroid constraint. We prove that, like the classical greedy algorithm, {\RRGreedy} achieves $\nicefrac{1}{2}$-approximation, but, more importantly, it also has additional useful properties that the classical greedy algorithm does not have.

Our main algorithm combines the two above mentioned components in a very natural way. Specifically, it uses the split algorithm to get two disjoint sets $A_1$ and $B_1$ whose union is a base, and it then completes each one of these sets into a full base by running {\RRGreedy} on appropriate contracted matroids.
The algorithm then outputs the better among these two bases, and we show that for a certain choice of the parameters its approximation ratio is strictly better than $\nicefrac{1}{2}$. Moreover, the only random part of our algorithm is {\RRGreedy}, which can be derandomized, making our algorithm fully deterministic.

\subsection{Additional Related Results}

The best known deterministic algorithms for most submodular maximization problems are not as good as their randomized counterparts. One exception is the problem of unconstrained maximization of a non-monotone submodular function. For this problem \cite{BF18} obtained a deterministic $\nicefrac{1}{2}$-approximation, which exactly matches the approximation ratio of the best randomized algorithm~\cite{BFNS15} and the information theoretic upper bound~\cite{FMV11}.
In other problems there is usually a gap between the performance guarantees of the best known deterministic and randomized algorithms. For example, 
for the problem of maximizing a \emph{non-monotone} submodular function subject to a matroid constraint the best deterministic algorithm has an approximation ratio of $\nicefrac{1}{4}$~\cite{LMNS10}, while the best randomized algorithm known achieves $0.385$-approximation~\cite{BF16a} (the best hardness result known for this problem is $0.478$ due to~\cite{GV11}). Another example is the problem of maximizing a monotone submodular function subject to a constant number of packing and covering constraints. Recently, Mizrachi et al.~\cite{MSSU18} developed a randomized $1 - 1/e - \eps$ approximation algorithm for this problem (this is optimal up to the $\eps$ term~\cite{NW78}). In the same work Mizrachi et al. also present two deterministic algorithms, but these algoritms only apply to special cases of the problem and achieve much worse approximation ratios (the better of which is $\nicefrac{1}{e}$). Nevertheless, despite these examples, we are not aware of any result showing a provable gap, for any submodular maximization problem, between the approximation ratios that can be achieved by deterministic and randomized algorithms.

The submodular welfare problem is an important special case of the problem of maximizing a monotone submodular function over a matroid. For this problem the Random Residual Greedy algorithm ({\RRGreedy}) is especially interesting because it is equivalent to applying the greedy algorithm to a natural online version of this problem under a random arrival model. Korula et al.~\cite{KMZ15} were interested in this online version and proved that {\RRGreedy} achieves $0.5052$-approximation for it. Recently, we were able to further improve the bound on the performance of {\RRGreedy} to $0.5096$ in the more general setting of partition matroids~\cite{BFG18} (this setting also allows one to interpret {\RRGreedy} as applying the greedy algorithm to a random order online variant of the problem).

Recall that the algorithm of C\u{a}linescu et al.~\cite{CCPV11} for maximizing a monotone submodular function subject to a matroid constraint achieves an optimal randomized approximation ratio of $1 - \nicefrac{1}{e}$. Despite this optimality, many works have aimed to improve over this algorithm in various aspects. The current work can be viewed as a first step towards a possible deterministic algorithm achieving the same approximation ratio, and the work of Filmus and Ward~\cite{FW14} mentioned earlier aimed to achieve the same randomized approximation ratio with a more combinatorial algorithm. Another line of work was dedicated to improving over the algorithm of~\cite{CCPV11} in terms of the time and oracle complexity~\cite{BV14,BFS17,MBKVK15}.

%% file: Preliminaries.tex
\section{Preliminaries} \label{sec:preliminaries}

We begin this section by formally defining the terms used in Section~\ref{sec:introduction} and the notation that we use in the next sections.

For every two sets $S, T \subseteq \cN$, we denote the marginal contribution of adding $T$ to $S$, with respect to a set function $f\colon 2^\cN \to \bR$, by $f(T \mid S) \triangleq f(T \cup S) - f(S)$. For an element $u \in \cN$ we use $f(u \mid S)$, $S + u$ and $S - u$ as shorthands for $f(\{u\} \mid S)$, $S \cup \{u\}$ and $S \setminus \{u\}$, respectively. We say that the set function $f$ is \emph{monotone} if $f(S) \leq f(T)$ for every two sets $S \subseteq T \subseteq \cN$ and {\em submodular} if $f(u \mid S) \geq f(u \mid T)$ for every two such sets and an element $u \in \cN \setminus T$. Since the description of a submodular function might be exponential in the size of the ground set, it is customary in the field of submodular maximization to assume that algorithms have access to the objective function $f$ only through a value oracle, \ie, an oracle that given a set $S \subseteq \cN$ returns $f(S)$.

A matroid $\cM$ over a ground set $\cN$ is defined as a pair $(\cN, \cI)$, where $\cI \subseteq 2^\cN$ obeys three properties: (i) $\cI \neq \varnothing$, (ii) if $S \subseteq T \subseteq \cN$ and $T \in \cI$, then $S \in \cI$ and (iii) if $S, T \in \cI$ and $|S| < |T|$ then there is an element $u \in T \setminus S$ such that $S + u \in \cI$. The sets in $\cI$ are called the \emph{independent} sets of $\cM$, and they are the feasible sets according to the constraint corresponding to this matroid. A base of a matroid is an inclusion-wise maximal independent set. It is not difficult to argue that all bases of a matroid have the same size. This size is known as the rank of the matroid, and we denote it by $k$ throughout this paper. Additionally, given an independent set $A$ of $\cM$, we denote by $\cM /A$ the matroid obtained from $\cM$ by contracting $A$. We refer the reader to~\cite{S03} for more information about matroid theory, as we assume basic knowledge of this theory. Like in the case of submodular functions, the size of the description of a matroid can be exponential in the size of its ground set. Thus, it is customary to assume that algorithms have access to matroids only through an independence oracle that given a set $S \subseteq \cN$ answers whether $S$ is independent or not.

Now, using the definitions given above, we formally state the problem that we consider in this paper. We are interested in the problem of maximizing a non-negative monotone submodular function $f\colon 2^\cN \to \nnR$ subject to a matroid $\cM = (\cN, \cI)$ constraint. For simplicity, we assume in the remaining parts of the paper that the rank $k$ of $\cM$ is at least $2$. Note that for $k = 1$ the above problem can be optimally solved by exhaustive search in linear time.

Next, we present a few technical lemmata that we use. We begin with (a rephrased version of) a useful lemma about submodular functions that was first proved in~\cite{FMV11}.

\begin{lemma}[Lemma~2.2 of~\cite{FMV11}] \label{lem:sampling}
Let $f\colon 2^\cN \to \bR$ be a submodular function, and let $T$ be an arbitrary set $T \subseteq \cN$. For every random set $T_p \subseteq T$ which contains every element of $T$ with probability $p$ (not necessarily independently),
\[
	\bE[f(T_p)] \geq (1 - p) f(\varnothing) + p \cdot f(T)
	\enspace.
\]
\end{lemma}

Next, we present (rephrased versions of) two known structural lemmata about matroids.

\begin{lemma}[Proved by~\cite{G73,W74}] \label{le:bases_exchange}
Given two bases $B_1$ and $B_2$ of a matroid $\cM$, and a partition $B_1 = X_1 \cupdot Y_1$, there is a partition $B_2 = X_2 \cupdot Y_2$ such that $X_1 \cupdot Y_2$ and $X_2 \cupdot Y_1$ are both bases of $\cM$.
\end{lemma}


\begin{lemma}[Proved by~\cite{B69} and can also be found as Corollary~39.12a in~\cite{S03}] \label{le:perfect_matching_two_bases}
Let $A$ and $B$ be two bases of a matroid $\cM = (\cN, \cI)$. Then, there exists a bijection $h : A \setminus B \rightarrow B \setminus A$ such that for every $u \in A \setminus B$, $(B - h(u)) + u \in \cI$.
\end{lemma}

Finally, we prove in Appendix~\ref{sec:proof} the following additional lemma about matroids, which generalizes the last lemma.

\begin{lemma}\label{lem:matroid_matching_greedy}
Let $A$ and $B$ be two bases of a matroid $\cM = (\cN, \cI)$, where $A$ is a maximum weight base according to some weight function $w\colon \cN \to \nnR$.
Then, there exist a bijective function $h\colon A \to B$ such that for every element $u \in A$
\begin{enumerate}
	\item $(B - h(u)) + u$ is a base of $\cM$.
	\item $w(u) \geq w(h(u))$.
\end{enumerate}
\end{lemma}

\paragraph{Paper Organization.} In Section~\ref{sec:basic_algorithms}, we present the two simple algorithms that are used as building blocks for our main algorithm. In Section~\ref{sec:main_algorithm}, we present the main algorithm itself, and finally, in Section~\ref{sec:derandomizing_greedy}, we explain how to derandomize one of the two algorithms presented in Section~\ref{sec:basic_algorithms} (the other one is deterministic to begin with).

%% file: basic-algorithms.tex
\section{Basic Algorithms} \label{sec:basic_algorithms}


\subsection{Split Algorithm}

The first of the above mentioned simple algorithms that are used as building blocks for our main algorithm is given as Algorithm~\ref{alg:split}. We note that this algorithm takes a parameter $p \in [0, 1]$ as input. 

\begin{algorithm}
\caption{\Split$(f, \cM, p)$} \label{alg:split}
Initialize: $A_0 \leftarrow \varnothing$, $B_0 \leftarrow \varnothing$.\\
\For{$i$ = $1$ \KwTo $k$}
{
    Let $u^A_i = \arg\max_{u \in \cM / (A_{i - 1} \cup B_{i-1})}\{f(u \mid A_{i - 1})\}$.\\
    Let $u^B_i = \arg\max_{u \in \cM / (A_{i - 1} \cup B_{i-1})}\{f(u \mid B_{i - 1})\}$.\\
    \If{ $p \cdot f(u^A_i \mid  A_{i - 1}) \geq (1-p) \cdot f(u^B_i \mid  B_{i - 1})$ }
    {
    $A_i \leftarrow A_{i -1} + u^A_i$.}
    \Else
    {$B_i \leftarrow B_{i -1} + u^B_i$.
    }
}
\Return{$(A_k,B_k)$}.
\end{algorithm}

The rest of this section is devoted to analyzing Algorithm~\ref{alg:split}. We begin with the following immediate observation.

\begin{observation}
The output sets $A_k$ and $B_k$ of Algorithm~\ref{alg:split} are disjoint, and their union is a base of $\cM$.
\end{observation}

Our next objective is to lower bound the values of the output sets of Algorithm~\ref{alg:split}.


\begin{lemma}\label{lem:split}
Let $T$ be a base of $\cM$ and $\frac{1}{5} \leq \beta \leq \frac{4}{5}$, then for $p=\frac{\beta}{\beta + \sqrt{(1-\beta)\beta}}$, Algorithm \ref{alg:split} satisfies
\[
	\beta \cdot f(A_k) + (1-\beta)\cdot f(B_k) \geq \frac{2}{3}\left(1- \sqrt{(1-\beta)\beta}\right) \cdot f(T)
	\enspace.
\]
\end{lemma}


\begin{proof}

By construction of the algorithm, we get for every $\alpha \in [0, 1]$

\begin{align} \label{eq:p_combination}
	p \cdot [f(A_i)- f(A_{i - 1})] + (1-p) \cdot [f(B_i)- f(B_{i - 1})&]\\ \nonumber
	={}&
	\max \mleft\{p \cdot f(u_i^A \mid A_{i-1}), (1-p) \cdot f(u_i^B \mid B_{i-1})\mright\}\\ \nonumber
	\geq{} & \alpha p \cdot f(u_i^A \mid A_{i-1}) + (1-\alpha)(1 - p) \cdot f(u_i^B \mid B_{i-1})
	\enspace.
\end{align}

Let us now construct for every $0 \leq i \leq k$ a set $\partT{T}{i}$ such that $A_i \cupdot B_i \cupdot \partT{T}{i}$ is a base of $\cM$. For $i=0$, we define $\partT{T}{0} = T$, and for $0 < i \leq k$ the set $\partT{T}{i}$ is defined recursively based on the behavior of Algorithm~\ref{alg:split} as follows.
Assume that $\partT{T}{i-1}$ is already constructed, and let $u_i$ denote the single element of $(A_i \cupdot B_i) \setminus (A_{i-1} \cupdot B_{i-1})$---\ie, the element that was added by Algorithm \ref{alg:split} in its $i$-th iteration. Since $A_i \cupdot B_i = A_{i-1} \cupdot B_{i-1} + u_i$ is an independent set of $\cM$ and $A_{i-1} \cupdot B_{i-1} \cupdot T_{i-1}$ is a base of $\cM$, there must be an element $v_i \in T_{i-1}$ such that $A_i \cupdot B_i \cupdot (T_{i-1} - v_i)$ is a base of $\cM$. Setting now $T_i = T_{i-1} - v_i$, we are guaranteed that $A_i \cupdot B_i \cupdot \partT{T}{i}$ is a base of $\cM$ as required.

Consider now an arbitrary $1 \leq i \leq k$. Since $A_{i - 1} \cupdot B_{i - 1} \cupdot \partT{T}{i-1}$ is a base of $\cM$, $v_i$ is a candidate for both $u_i^A$ and $u_i^B$. Together with the fact that $u_i^A$ and $u_i^B$ are maximizers with respect to $f(\cdot \mid A_{i-1})$ and $f(\cdot \mid B_{i-1})$, we get that the rightmost expression of~\eqref{eq:p_combination} is at least
\begin{align*}
	\alpha p \cdot f(&v_i \mid A_{i-1}) + (1-\alpha) (1-p) \cdot f(v_i \mid B_{i-1})\\
	\geq{} &
	\alpha p \cdot f(v_i \mid A_{i-1} \cupdot T_i) + (1-\alpha) (1-p) \cdot f(v_i \mid B_{i-1} \cupdot T_i)\\
	={} &
	\alpha p \cdot [f(A_{i - 1} \cup T_{i-1}) - f(A_{i-1} \cup T_i)] + (1-\alpha) (1-p) \cdot [f(B_{i - 1} \cup T_{i-1}) - f(B_{i-1} \cup T_i)]\\
	\geq{} &
	\alpha p \cdot [f(A_{i - 1} \cup T_{i-1}) - f(A_i \cup T_i)] + (1-\alpha)(1-p) \cdot [f(B_{i - 1} \cup T_{i-1}) - f(B_{i} \cup T_i)]
	\enspace,
\end{align*}
where the first inequality follows by submodularity, and the final inequality follows by monotonicity.
Combining the last two inequalities and rearranging, we get that the expression
\[p \cdot f(A_i) +p\alpha \cdot f(A_{i} \cup T_{i}) + (1-p) \cdot f(B_i) + (1-p)(1-\alpha) \cdot f(B_{i} \cup T_{i})\] is an increasing function of $i$ for $0\leq i\leq k$.
In particular, since $A_0 = B_0 = T_k = \varnothing$ and $T_0 = T$, we get
\begin{align*}
	p(1+\alpha) \cdot f(A_k)  + (1-p)(2-\alpha) \cdot f(B_k)
	\geq{} &
	f(\varnothing) + [\alpha p + (1-\alpha)(1-p)] \cdot f(T)\\
	\geq{} &
	[\alpha p + (1-\alpha)(1-p)] \cdot f(T)
	\enspace,
\end{align*}
where the second inequality follows from the non-negativity of $f$.

We now set $p= \frac{\beta}{\beta + \sqrt{(1-\beta)\beta}}$ and $\alpha=\frac{1+\beta-3\sqrt{(1-\beta)\beta}}{2\beta- 1}$
(for $\beta=\nicefrac{1}{2}$, this expression for $\alpha$ is not defined, so we set $\alpha=\nicefrac{1}{2}$).
Note that for $\frac{1}{5} \leq \beta \leq \frac{4}{5}$ these values for $p$ and $\alpha$ are indeed in the range $[0,1]$, and additionally they imply $p(1+\alpha)= \frac{3 \beta}{1 + 2 \sqrt{(1 - \beta) \beta}}$, $(1-p)(2-\alpha) = \frac{3-3\beta}{1+2\sqrt{(1-\beta)\beta}}$ and
$p\alpha + (1-p)(1 - \alpha) = \frac{2(1 - \sqrt{(1-\beta)\beta})}{1+2\sqrt{(1-\beta)\beta}}$. The lemma now follows by plugging these expressions into the previous inequality and multiplying by $(1 + 2 \sqrt{(1 - \beta) \beta})/3$.
%
\end{proof}

The final property of Algorithm~\ref{alg:split} that we need to prove is that for every base $T$ of $\cM$ there is a good way to split $T$ with respect to the output sets of the algorithm.

\begin{lemma} \label{lem:split_partition}
For every base $T$ of $\cM$, there exists a partition of $T$ into two disjoint sets $T_A \cupdot T_B$ such that
\begin{compactitem}
	\item $A_k \cupdot T_A$ and $B_k \cupdot T_B$ are both bases of $\cM$.
	\item $f(A_k) + f(A_k \cup T_A) \geq f(T)$ and $f(B_k) + f(B_k \cup T_B) \geq f(T)$.
\end{compactitem}
\end{lemma}
\begin{proof}
By Lemma~\ref{le:bases_exchange}, since $A_k \cupdot B_k$ is a base of $\cM$, there must be a partition of $T$ into two disjoint subsets $T_A$ and $T_B$ such that $A_k \cupdot T_A$ and $B_k \cupdot T_B$ are both bases. In the remaining part of the proof we show that $f(A_k) + f(A_k \cup T_A) \geq f(T)$. Proving that the inequality $f(B_k) + f(B_k \cup T_B) \geq f(T)$ also holds can be done in a symmetric way.

We prove by induction that for every $0 \leq i \leq k$, there must exist a set $\partT{T}{i} \subseteq T_B$ such that $A_i \cupdot B_k \cupdot \partT{T}{i}$ is a base of $\cM$ and $f(A_i) + f(A_i \cupdot \partT{T}{i} \cupdot T_A) \geq f(T)$. For $i = 0$ we define $T_0 = T_B$, which makes the claim that we would like to prove by induction trivial since $f$ is non-negative and $A_0 = \varnothing$. Assume now that this claim holds for $0 \leq i - 1 < k$, and let us prove it for $i$. There are two cases to consider. If $A_i = A_{i-1}$, then we are done due to the induction hypothesis by setting $T_i = T_{i-1}$. Thus, it remains to consider the case in which $A_i = A_{i-1} + u_i^A$. In this case, since $A_i \cupdot B_k \subseteq A_k \cupdot B_k$ is independent in $\cM$ and $A_{i - 1} \cupdot B_k \cupdot T_{i-1}$ is a base, there must be an element $v_i \in T_{i-1}$ such that $A_i \cupdot B_k \cupdot (T_{i-1} - v_i) = (A_{i-1} + u_i) \cupdot B_k \cupdot (T_{i-1} - v_i)$ is also a base of $\cM$. Choosing now $T_i = T_{i-1} - v_i$, which certainly obeys the requirement that $A_i \cupdot B_k \cupdot T_i$ is a base of $\cM$, we get
\begin{align*}
	f(A_i) - f(&A_{i-1})
	=
	f(u_i^A \mid A_{i-1})
	\geq
	f(v_i \mid A_{i-1})
	\geq
	f(v_i \mid A_{i-1} \cupdot T_i \cupdot T_A)\\
	={} &
	f(A_{i-1} \cupdot T_{i-1} \cupdot T_A) - f(A_{i-1} \cupdot T_i \cupdot T_A)
	\geq
	f(A_{i-1} \cupdot T_{i-1} \cupdot T_A) - f(A_i \cupdot T_i \cupdot T_A)
	\enspace,
\end{align*}
where the first inequality follows from the choice of $u_i^A$ since the fact that $A_{i-1} \cupdot B_k \cupdot T_{i-1}$ is independent implies that $v_i$ is a candidate for $u_i^A$, the second inequality follows from the submodularity of $f$ and the last inequality follows from its monotonicity. Combining the last inequality with the induction hypothesis, we now get
\[
	f(A_i) + f(A_i \cupdot T_i \cupdot T_A)
	\geq
	f(A_{i-1}) + f(A_{i - 1} \cupdot T_{i - 1} \cupdot T_A )
	\geq
	f(T)
	\enspace,
\]
which completes the proof by induction.

Plugging $i = k$ into the claim proved above, and observing that the fact that $A_k \cupdot B_k \cupdot T_k$ is a base of $\cM$ implies $T_k = \varnothing$, we get
\[
	f(A_k) + f(A_k \cupdot T_A)
	\geq
	f(T)
	\enspace.
	\qedhere
\]
\end{proof}


\subsection{Residual Random Greedy Algorithm}

The second simple algorithm that we need is a procedure known as the Residual Random Greedy algorithm (\RRGreedy) that was originally described by~\cite{BFNS14} and is given here as Algorithm~\ref{alg:ResidualRandomGreedy}.

\begin{algorithm}
\caption{\textsf{Residual Random Greedy -- \RRGreedy}$(f,\cM)$} \label{alg:ResidualRandomGreedy}
Initialize: $A_0 \leftarrow \varnothing$.\\
\For{$i$ = $1$ \KwTo $k$}
{
    Let $M_i$ be a base of $\cM / A_{i-1}$ maximizing $\sum_{u \in M_i} f(u \mid A_{i - 1})$.\\
    Let $A_i \leftarrow A_{i -1} + u_i$, where $u_i$ is a uniformly random element from $M_i$.\\
}
Return $A_{k}$.
\end{algorithm}

For the analysis of Algorithm~\ref{alg:ResidualRandomGreedy}, we use the following construction.
Let $T$ be an arbitrary base of $\cM$. Then, we construct for every $0 \leq i \leq k$ a set $\partT{T}{i}$ which is a base of $\cM / A_i$ as follows.
We define $\partT{T}{0} = T$, and for $0 < i \leq k$ we define $\partT{T}{i}$ recursively based on the behavior of Algorithm~\ref{alg:ResidualRandomGreedy}. Assume $\partT{T}{i-1}$ is already constructed, and let $h_i\colon M_i \to \partT{T}{i - 1}$ be a bijection mapping every element $u \in M_i$ to an element of $\partT{T}{i - 1}$ in such a way that $(\partT{T}{i - 1} - h_i(u)) + u_i$ is a base of $\cM / A_{i-1}$. The existence of such a function follows immediately from Lemma~\ref{le:perfect_matching_two_bases} since $\partT{T}{i - 1}$ and $M_i$ are both bases of $\cM / A_{i - 1}$ ($h_i$ maps elements of $\partT{T}{i - 1} \cap M_i$ to themselves). We now set $\partT{T}{i} = \partT{T}{i - 1} - h_i(u_i)$, and one can observe that it is indeed a base of $\cM / A_i$ since $A_i = A_{i - 1} + u_i$. It is important for the analysis of Algorithm~\ref{alg:ResidualRandomGreedy} that the choice of $h_i$ (among the possibly multiple functions obeying the required properties) is made independently of the random choice of $u_i$ out of $M_i$. Note that choosing $h_i$ in such a way gurantees that $h_i(u_i)$ is a uniformly random element of $\partT{T}{i - 1}$, and thus implies the next observation.

\begin{observation} \label{obs:T_i_distribution}
$\partT{T}{i}$ is a uniformly random subset of $T$ of size $k-i$.
\end{observation}

The following lemma is a central component used in the proofs of all the claims that we present later regarding Algorithm~\ref{alg:ResidualRandomGreedy}.

\begin{lemma} \label{lem:A_OPT_development}
For every $1 \leq i \leq k$ and a (possibly random) set $S \subseteq \cN$,
\[
	\bE[f(A_i) + f(A_i \cup \partT{T}{i} \cup S)]
	\geq
	\bE[f(A_{i - 1}) + f(A_{i - 1} \cup \partT{T}{i - 1} \cup S)]
	\enspace.
\]
\end{lemma}

\begin{proof}
We prove that the lemma holds when conditioned on any fixed choice for the random decisions made by Algorithm~\ref{alg:ResidualRandomGreedy} in its first $i - 1$ iterations, which implies that the lemma holds also unconditionally by the law of total expectation. Given  such a conditioning, the sets $A_{i-1}$ and $M_i$ become deterministic, and thus, when implicitly assuming such a conditioning, we get
\begin{align*}
	\bE[f(A_i) - f(&A_{i-1})]
	=
	\bE[f(u_i \mid  A_{i-1})]
	=
	\frac{\sum_{u \in M_i} f(u \mid A_{i-1})}{k - i + 1}\\
	\geq{} &
	\frac{\sum_{u \in \partT{T}{i-1}} f(u \mid A_{i-1})}{k - i + 1}
	=
	\bE[f(h_i(u_i) \mid A_{i-1})]
	\geq
	\bE[f(h_i(u_i) \mid A_{i-1} \cup \partT{T}{i} \cup S)]\\
	={} &
	\bE[f(A_{i-1} \cup \partT{T}{i-1} \cup S) - f(A_{i-1} \cup \partT{T}{i} \cup S)]
	\geq
	\bE[f(A_{i-1} \cup \partT{T}{i-1} \cup S) - f(A_{i} \cup \partT{T}{i} \cup S)]
	\enspace,
\end{align*}
where the first inequality follows from the definition of $M_i$, the second inequality holds due to the monotonicity and submodularity of $f$, and the final inequality follows again from the monotonicity of $f$.
\end{proof}

\begin{corollary}\label{cor:mainineq}
For every $1 \leq i \leq k$ and a (possibly random) set $S \subseteq \cN$,
\[
	\bE[f(A_k) + f(A_k \cup S)]
	\geq
	\bE[f(A_i) + f(A_i \cup \partT{T}{i} \cup S)]
	\geq
	\bE[f(T \cup S)]
	\enspace.
\]
\end{corollary}
\begin{proof}
Lemma~\ref{lem:A_OPT_development} shows that the expectation of $f(A_i) + f(A_i \cup \partT{T}{i} \cup S)$ is a non-decreasing function of $i$. Thus,
\[
	\bE[f(A_k) + f(A_k \cup \partT{T}{k} \cup S)]
	\geq
	\bE[f(A_i) + f(A_i \cup \partT{T}{i} \cup S)]
	\geq
	\bE[f(A_0) + f(A_0 \cup \partT{T}{0} \cup S)]
	\enspace.
\]
The corollary now follows by recalling that $A_0 = \varnothing$ and $\partT{T}{0} = T$ by definition, observing that $f(A_0) \geq 0$ since $f$ is non-negative and observing that $T_k = \varnothing$ since $A_k$ and $A_k \cupdot T_k$ are both bases of $\cM$.
\end{proof}

Setting $S = \varnothing$, the last corollary implies that the expected value of $f(A_k)$ is at least half of $f(T)$. The next lemma gives a lower bound on $\bE[f(A_i)]$ that applies for other values of $i$ as well. Let $g(x) \triangleq x - x^2/2$.

\begin{lemma} \label{lem:basic_analysis}
For every $0 \leq i \leq k$,
$
	\bE[f(A_i)]
	\geq
	[g(\nicefrac{i}{k}) + \delta] \cdot f(T)
$ where $\delta = 1/(2k^2)$ for $0 < i < k$ and $0$ otherwise.
\end{lemma}

\begin{proof}
Since $g(1) = 1/2$, the above discussion implies that in the special case of $i = k$ the lemma follows from Corollary~\ref{cor:mainineq}. We prove the lemma for the other cases by induction. For $i = 0$ the lemma holds, even without the expectation, due to the non-negativity of $f$ since $g(0) = 0$. The rest of the proof is devoted to showing that the lemma holds for $1 \leq i < k$ given that it holds for $i - 1$.

Let $\cE$ be an arbitrary event fixing the random choices made by Algorithm~\ref{alg:ResidualRandomGreedy} in its first $i - 1$ iterations. Observe that conditioned on this event the sets $A_{i-1}$, $M_i$ and $\partT{T}{i-1}$ become deterministic. Thus, conditioned on $\cE$,
\begin{align*}
	\bE[f(A_i) - f(A_{i - 1})]
	={} &
	\bE[f(u_i \mid  A_{i-1})]
	=
	\frac{\sum_{u \in M_i} f(u \mid A_{i-1})}{k - i + 1}
	\geq
	\frac{\sum_{u \in \partT{T}{i}} f(u \mid A_{i-1})}{k - i + 1}\\
	\geq{} &
	\frac{f(\partT{T}{i} \mid A_{i-1})}{k - i + 1}
	\geq
	\frac{f(T) - 2\bE[f(A_{i-1})]}{k - i + 1}
	\enspace,
\end{align*}
where the first inequality follows from the definition of $M_i$, the second inequality follows from the submodularity of $f$ and the last inequality follows from the second inequality of Corollary~\ref{cor:mainineq} by choosing $S = \varnothing$.

Taking expectation now over all the possible choices of $\cE$, we get
{\mleftright
\begin{align*}
	\bE[f(A_i)&]
	\geq
	\bE[f(A_{i-1})] + \frac{f(T) - 2\bE[f(A_{i-1})]}{k - i + 1}
	=
	\frac{k - i - 1}{k - i + 1} \cdot \bE[f(A_{i-1})] + \frac{f(T)}{k - i + 1}\\
	\geq{} &
	\frac{k - i - 1}{k - i + 1} \cdot g\left(\frac{i - 1}{k}\right) \cdot f(T) + \frac{f(T)}{k - i + 1}
	=
	\left[g\left(\frac{i - 1}{k}\right) + \frac{1 - 2g(\frac{i-1}{k})}{k - i + 1} \right] \cdot f(T)
	\enspace,
\end{align*}
}%
where the second inequality follows from the induction hypothesis (since $i \leq k - 1$). Using the observation that the derivative $g'(x)$ of $g(x)$ obeys $g'(x) = (1 - 2g(x)) / (1 - x)$, the last inequality yields
{\mleftright
\begin{align*}
	\frac{\bE[f(A_i)]}{f(T)}
	\geq{} &
	g\left(\frac{i - 1}{k}\right) + \frac{1 - 2g(\nicefrac{(i-1)}{k})}{k - i + 1}
	=
	g\left(\frac{i - 1}{k}\right) + \frac{g'(\nicefrac{(i-1)}{k})}{k}\\
	={} &
	g\left(\frac{i - 1}{k}\right) + \int_{(i - 1)/k}^{i/k} g'(x) dx + \int_{(i - 1)/k}^{i/k} \left[g'\left(\frac{i-1}{k}\right) - g'(x)\right] dx\\
	={} &
	g(\nicefrac{i}{k}) + \int_{(i - 1)/k}^{i/k} \left(x - \frac{i-1}{k}\right) dx
	=
	g(\nicefrac{i}{k}) + \frac{1}{2k^2}
	\enspace.
	\qedhere
\end{align*}
}%
\end{proof}

The following lemma generalizes the previous one to non-integer values (at the cost of a small loss in the guarantee).
\begin{lemma} \label{lem:basic_analysis frac}
For every $0 \leq x < 1$, let $\alpha= \lfloor kx + 1 \rfloor - kx$ and $i=\lfloor xk \rfloor$, then
\[
	\alpha \cdot \bE[f(A_{i})]+ (1-\alpha)\cdot \bE[f(A_{i+1})]
	\geq g(x)  \cdot f(T)
	\enspace.
\]
\end{lemma}

\begin{proof}
If $xk$ is an integer, then the lemma follows directly from Lemma~\ref{lem:basic_analysis} (note that $\alpha = 1$ in this case).
Otherwise, we have $\alpha \cdot i + (1-\alpha)(i+1)= kx$.
Since we assume that $k \geq 2$, at least one of the values $i=\lfloor xk \rfloor$ or $i+1 = \lceil kx \rceil$ must belong to $\{1, \ldots,k-1\}$. Thus, by Lemma~\ref{lem:basic_analysis},
\begin{align*}
	\mspace{117mu}&\mspace{-117mu}\frac{\alpha \cdot \bE[f(A_{i})] + (1 - \alpha) \cdot \bE[f(A_{i+1})]}{f(T)}
	\geq
	\alpha \cdot g(\nicefrac{i}{k}) + (1-\alpha) \cdot g(\nicefrac{(i+1)}{k}) + \min\{\alpha, 1 - \alpha\} / (2k^2)\\
	={} &
	\alpha\left(\frac{i}{k} - \frac{i^2}{2k^2}\right) + (1-\alpha)\left(\frac{i+1}{k} - \frac{(i+1)^2}{2k^2}\right) + \frac{\min\{\alpha, 1 - \alpha\}}{2k^2}\\
	={} &
	\frac{\alpha i + (1 - \alpha)(i+1)}{k} - \frac{(\alpha i + (1 - \alpha) (i+1))^2}{2k^2} - \frac{\alpha(1 - \alpha)}{2k^2} + \frac{\min\{\alpha, 1 - \alpha\}}{2k^2}
\geq{} g(x)
	\enspace,
\end{align*}
where the final inequality follows since $\alpha i + (1 - \alpha)(i+1)=kx$ and $\min\{\alpha, 1 - \alpha\} \geq \alpha(1-\alpha)$.
\end{proof}

The next lemma uses the previous lemma to derive an additional lower bound on the value of the output set of Algorithm~\ref{alg:ResidualRandomGreedy}. 

\begin{lemma} \label{cor:combined_bound}
For every $0 \leq x \leq 1$ and base $T'$ of $\cM$, \[3\bE[f(A_k)] \geq (1 + g(x)) \cdot f(T') + (1 - x) \cdot f(T \mid T')\enspace.\]
\end{lemma}

\begin{proof}
We begin the proof by showing a lower bound on the gain of Algorithm~\ref{alg:ResidualRandomGreedy} during its last iterations. Observe that for every $i\in \{0, \ldots, k\}$ it holds that
\begin{align*}
	\bE[f(A_k \mid A_i)]
	\geq{} &
	\bE[f(\partT{T}{i} \mid A_k \cup T')]
	=
	\bE[f(\partT{T}{i} \cup A_k \mid T') - f(A_k \mid T')]
	\geq
	\bE[f(\partT{T}{i} \mid T')] - \bE[f(A_k \mid T')]\\
	={} &
	\bE[f(\partT{T}{i} \mid T')] - \bE[f(A_k \cup T')] + f(T')
	\geq
	\bE[f(\partT{T}{i} \mid T')] - 2\bE[f(A_k)] + f(T')
	\enspace,
\end{align*}
where the first inequality follows from the first inequality of Corollary~\ref{cor:mainineq} by plugging $S = A_k \cup T'$ and the second inequality follows from monotonicity. To see why the last inequality holds, observe that $T$ and $T'$ are both arbitrary bases of $\cM$, and thus Corollary~\ref{cor:mainineq} is still true even if we replace $T$ in its guarantee by $T'$. Setting now $S = A_k$ in this modified Corollary~\ref{cor:mainineq} implies the above last inequality.

Now, recall that $\partT{T}{i}$ is a random set that contains every element of $T$ with probability $1 - \nicefrac{i}{k}$. Thus, we can use Lemma~\ref{lem:sampling} to get
\[
	\bE[f(\partT{T}{i} \mid T')]
	\geq
	\nicefrac{i}{k} \cdot f(\varnothing \mid T') + (1 - \nicefrac{i}{k}) \cdot f(T \mid T')
	=
	(1 - \nicefrac{i}{k}) \cdot f(T \mid T')
	\enspace.
\]

For $x=1$ the lemma follows directly by Corollary~\ref{cor:mainineq} (for $S = \varnothing$). For $x < 1$, choosing $i=\lfloor kx \rfloor$ and $\alpha= \lfloor kx + 1 \rfloor - kx$ we get from the previous two inequalities
\begin{align*}
	\alpha \cdot \bE[f(A_k \mid A_{i})] + (1 - \alpha) \cdot \bE[&f(A_k \mid A_{i+1})]\\
	\geq{} &
	\alpha\mleft[(1 - \nicefrac{i}{k}) \cdot f(T \mid T') - 2\bE[f(A_k)] + f(T')\mright]\\
	&+
	(1 - \alpha)\mleft[(1 - \nicefrac{(i+1)}{k}) \cdot f(T \mid T') - 2\bE[f(A_k)] + f(T')\mright]\\
	={} &
	(1 - x) \cdot f(T \mid T') - 2\bE[f(A_k)] + f(T')
	\enspace.
\end{align*}
The lemma now follows 
by adding to the last inequality the inequality $\alpha \cdot \bE[f(A_{i})] + (1 - \alpha) \cdot \bE[f(A_{i+1})] \geq g(x) \cdot f(T')$, which holds by Lemma~\ref{lem:basic_analysis frac}, and rearranging. 
\end{proof}

For convenience, the following proposition summarizes the properties of Algorithm~\ref{alg:ResidualRandomGreedy} that we use in the analysis of our main algorithm.
\begin{proposition} \label{prop:random_greedy_properties}
Given bases $T_1$ and $T_2$ of $\cM$, the output set $A$ of Algorithm~\ref{alg:ResidualRandomGreedy} obeys
\begin{compactenum}
	\item $\bE[f(A)] \geq f(T_1)/2$. \label{item:value_A}
	\item $3\bE[f(A)] \geq (1 + g(x)) \cdot f(T_1) + (1 - x) \cdot f(T_2 \mid T_1)$ for every $x \in [0, 1]$. \label{item:composite_bound}
\end{compactenum}
\end{proposition}
\begin{proof}
The first part of the proposition follows from Lemma~\ref{lem:basic_analysis} by setting $T = T_1$ since $g(1) = 1/2$, and the second part follows from Corollary~\ref{cor:combined_bound} by setting $T = T_2$ and $T' = T_1$.
\end{proof} 

%% file: sub-monotone-matroid.tex
\section{Main Algorithm} \label{sec:main_algorithm}

In this section we present our main algorithm, which is given as Algorithm~\ref{alg:divide}. This algorithm invokes the basic algorithms presented in Section~\ref{sec:basic_algorithms}. Note that the invocation of the algorithm {\Split} requires a value for the parameter $p \in [0, 1]$ which is left unspecified by the pseudocode of Algorithm~\ref{alg:divide}, but is determined later in this section.

\begin{algorithm}
\caption{\textsf{Matroid Split and Grow}$(f,\cM)$} \label{alg:divide}
$(A_1,B_1) \gets \text{\Split}(f, \cM, p)$.\\
$A_2 \gets \mbox{\RRGreedy}(f(\cdot \mid A_1), \cM / A_1)$. \\
$B_2 \gets \mbox{\RRGreedy}(f(\cdot \mid B_1), \cM / B_1)$. \\
Return the better solution out of $A=(A_1 \cupdot A_2)$ and $B=(B_1 \cupdot B_2)$.
\end{algorithm}

Algorithm~\ref{alg:divide} constructs two solutions ($A$ and $B$) using a two steps process. In the first step it constructs two disjoint sets $A_1$ and $B_1$ whose union is a base of $\cM$, and in the second step it grows each one of these sets into a base. A central observation used in the analysis of the algorithm is that one way to grow $A_1$ into a base is to add $B_1$ to it and vice versa. A (potentially) different way to grow $A_1$ and $B_1$ into bases is by adding to them appropriate subsets $OPT_A$ and $OPT_B$ of $OPT$. Specifically, we use the partition of $OPT$ into two sets $OPT_A = T_A$ and $OPT_B = T_B$ whose existence is guaranteed by Lemma~\ref{lem:split_partition} when we set $T = OPT$.

The following lemma shows that the two sets that can complement $A_1$ into a base according to the above discussion (\ie, $B_1$ and $OPT_A$) have significant value together---unless the algorithm does very well in its attempt to grow $B_1$ into a base. 

\begin{lemma}\label{lem:b}
$f(B_1 \cup OPT_A) \geq f(OPT) - 2 \bE[f(B \mid B_1)]$.
\end{lemma}

\begin{proof}
Observe that
\begin{align*}
	f(B_1 \cup OPT_A)
	={}&
	f(OPT_A \mid B_1) + f(B_1)
	\geq
	f(OPT_A \mid B_1 \cup OPT_B) + f(B_1)\\
	\geq{} &
	f(OPT) - f(OPT_B \mid  B_1)
	\geq
	f(OPT) - 2\bE[f(B \mid B_1)]
	\enspace,
\end{align*}
where the first inequality follows by the submodularity of $f$, the second inequality follows by the monotonicity of $f$ and the last inequality follows by invoking the first part of Proposition~\ref{prop:random_greedy_properties} for the second execution of {\RRGreedy} with $T_1 = OPT_B$ (note that $OPT_B$ is indeed a base of the matroid $\cM / B_1$ passed to this execution).
\end{proof}


We can now get a lower bound on a linear combination of the values of the two solutions produced by Algorithm~\ref{alg:divide}.

\begin{lemma}\label{lem:guarantee}
For every $0\leq x\leq 1$,
\[
	3\bE[f(A)]+ 2(1-x) \cdot \bE[f(B)]
	\geq (1+g(x)) \cdot f(OPT) + (2 -x-2g(x)) \cdot f(A_1) + 2(1-x) \cdot f(B_1)
	\enspace.
\]

\end{lemma}

\begin{proof}
Consider the first execution of {\RRGreedy} invoked by Algorithm~\ref{alg:divide}. Since $OPT_A$ and $B_1$ are two bases of the matroid passed to this execution, using the second part of Proposition~\ref{prop:random_greedy_properties} with $T_1 = OPT_A$, $T_2 = B_1$, we get
\begin{align*}
	3\bE[f(A_2 \mid A_1)]
	\geq{} &
	(1 + g(x)) \cdot f(OPT_A \mid A_1) + (1 - x) \cdot f(B_1 \mid OPT_A \cup A_1)\\
	={} &
	(x + g(x)) \cdot f(A_1 \cup OPT_A) - (1 + g(x)) \cdot f(A_1) + (1 - x) \cdot f(A_1 \cup B_1 \cup OPT_A)
	\enspace.
\end{align*}

Let us now present lower bounds for two of the terms on the right hand side of the last inequality. The term $f(OPT_A \cup A_1)$ is at least $f(OPT) - f(A_1)$ by Lemma~\ref{lem:split_partition} and the definition of $OPT_A$. Additionally, by the monotonicity of $f$, the term $f(A_1 \cup B_1 \cup OPT_A)$ is at least $f(B_1 \cup OPT_A)$, and this last expression can be lower bounded by Lemma~\ref{lem:b}. Plugging these lower bounds into the last inequality, we get
\begin{align*}
	3\bE[f(A_2&{} \mid A_1)]\\
	\geq{} &
	(x + g(x)) \cdot [f(OPT) - f(A_1)] - (1 + g(x)) \cdot f(A_1) + (1 - x) \cdot \bE[f(OPT) - 2\cdot f(B \mid B_1)]\\
	={} &
	(1 + g(x)) \cdot f(OPT) - (1 + x + 2g(x)) \cdot f(A_1) - 2(1-x) \cdot \bE[f(B \mid B_1)]
	\enspace.
\end{align*}

Rearranging the terms, we get the desired.
\end{proof}

To get a lower bound on the competitive ratio of Algorithm~\ref{alg:divide} we need to plug into the guarantee of Lemma~\ref{lem:guarantee} a value for $x$ and to lower bound the terms that include $f(A_1)$ and $f(B_1)$ in this guarantee. This is done in the proof of the next proposition.

\begin{proposition}
The approximation ratio of Algorithm~\ref{alg:divide} is at least $0.5008$.
\end{proposition}
\begin{proof}
Let \[\beta = \frac{2-x-2g(x)}{(2 -x-2g(x)) + 2(1-x)} = \frac{2 -x-2g(x)}{4 -3x-2g(x)} \enspace.\] Plugging this value of $\beta$ into the guarantee of Lemma~\ref{lem:split} for $T = OPT$, and choosing the value of $p$ accordingly, we get
\[
	(2 -x-2g(x)) \cdot f(A_1) + 2(1-x)f(B_1)
	\geq
	(4 -3x-2g(x)) \cdot w(\beta) \cdot f(OPT)
	\enspace,
\]
where $w(\beta) \triangleq \frac{2}{3}\left(1- \sqrt{(1-\beta)\beta}\right)$.
Combining this inequality with the guarantee of Lemma~\ref{lem:guarantee}, we get
\begin{align*}
	\max\mleft\{\bE[f(A)],\bE[f(B)]\mright\}\geq{}&  \frac{3\bE[f(A)]+ 2(1-x) \cdot \bE[f(B)]}{5-2x}\\
	\geq{}& \frac{1+g(x) + (4 -3x-2g(x)) \cdot w(\beta)}{5-2x} \cdot f(OPT)
	\enspace.
\end{align*}

Setting $x=0.9$, the coefficient of $f(OPT)$ in the last inequality becomes larger than $0.5008$. Moreover, it can be verified that for this value of $x$, $\beta \approx 0.35$ which is in the range $[\frac{1}{5}, \frac{4}{5}]$, as required by Lemma~\ref{lem:split}.
%
%
%
%
\end{proof}

%% file: DeterministicGreedy.tex
\section{Derandomizing Algorithm~\ref{alg:ResidualRandomGreedy}} \label{sec:derandomizing_greedy}

In this section, we describe a deterministic algorithm (given as Algorithm~\ref{alg:deterministic_greedy}) whose output obeys (roughly) the same properties guaranteed by Proposition~\ref{prop:random_greedy_properties} for Algorithm~\ref{alg:ResidualRandomGreedy}, and thus, using it in Algorithm~\ref{alg:divide} instead of Algorithm~\ref{alg:ResidualRandomGreedy} does not affect the approximation guarantee of the former algorithm. 
Note that Algorithm~\ref{alg:deterministic_greedy} gets a parameter $B$ that Algorithm~\ref{alg:ResidualRandomGreedy} does not get. We assume that this parameter is a base of $\cM$.


\begin{algorithm}[th]
\caption{\textsf{Residual Parallel Greedy -- \RPGreedy}$(f,\cM, B)$} \label{alg:deterministic_greedy}
\DontPrintSemicolon
Initialize: $A^j_0  \leftarrow \varnothing$ and $B^j_0 \leftarrow B$ for every $j = 1, \ldots, k$. \\
\For{$i = 1$ \KwTo $k$}
{
For every $j=1, \ldots, k$, let $M^j_i$ be a base of $\cM / A^j_{i-1}$ maximizing $\sum_{u \in M^j_i} f(u \mid A^j_{i - 1})$.\\
Construct a weighted bipartite (multi-)graph $G_i = (V_{L}, V_{R}, E, w)$ as follows.\\
{\begin{itemize}
	\item $V_{L} \triangleq B$ and $V_{R}\triangleq\{1,\ldots, k\}$.
	\item For each $u\in  M^j_i$ and $v \in B$, add an edge $e=(v,j)$ with weight $w_e = f(u \mid A^j_{i-1})$ if
\begin{itemize}
\item $v \in B^j_{i-1}$, and $(A^j_{i-1} + u) \cupdot (B^j_{i-1} - v)$ is a base of $\cM$.
\item $f(u \mid A^j_{i-1}) \geq f(v \mid A^j_{i-1})$.
\end{itemize}
\end{itemize}}
Find a maximum weight perfect matching $R_i$ in $G_i$.\label{line:perfect_matching}\\
	\For{every $j = 1$ \KwTo $k$}
	{
		Let $e=(v_i^j,j)$ be the single edge in the matching $R_i$ which hits $j$, and let $u_i^j\in M_i^j$ be the element that corresponds to this edge.\\
		Set $A^j_i \gets A^j_{i -1} + u_i^j$ and $B^j_i \gets B^j_{i-1} - v_i^j$.\\
	}
}
\Return{the best set out of $A^1_{k}, A^2_{k}, \dotsc, A^k_k$}.
\end{algorithm}


We begin the analysis of Algorithm~\ref{alg:deterministic_greedy} with the following lemma which guarantees, in particular, that the algorithm can always find a perfect matching in $G_i$.
\begin{lemma} \label{lem:basic_deterministic_properties}
The algorithm satisfies the following properties.
\begin{itemize}
\item For every $i=0, \ldots, k$ and $j=1, \ldots, k$, $A_i^j \cupdot B^j_i$ is a base of $\cM$ and every element $u \in B$ appears in exactly $k - i$ out of the sets $B^1_i, B^2_i, \dotsc, B^k_i$.
\item For every $i=1, \ldots, k$, $G_i$ has a perfect matching of weight at least
\[
	\frac{1}{k - i + 1} \cdot \sum_{j = 1}^k \sum_{u \in M_i^j} f(u \mid A_{i-1}^j)
	\enspace.
\]
\end{itemize}
\end{lemma}
\begin{proof}
We prove the lemma by induction on $i$. For $i = 0$ the lemma is trivial since, for every $1 \leq j \leq k$, $A^j_0 = \varnothing$ and $B^j_0 = B$ is a base of $\cM$ by definition. Thus, it remains to prove the lemma for $1 \leq i \leq k$ under the assumption that it holds for $i - 1$.

For every $1 \leq j \leq k$, $M^j_i$ and $B^j_{i - 1}$ are both bases of $\cM / A^j_{i-1}$, and $M^j_i$ maximizes the linear function $w(u) = f(u \mid A^j_i)$ among all such bases.
Thus, by Lemma~\ref{lem:matroid_matching_greedy}, there exists a bijective function $h^j\colon M^j_i \to B^j_{i - 1}$ such that $(B^j_{i - 1} - h^j(u)) + u$ is a base of $\cM / A^j_{i - 1}$ and $f(u \mid A^j_{i-1}) \geq f(h^j(u) \mid A^j_{i-1})$ for every $u \in M^j_i$. 
For every $u \in M^j_i$, let $e(j,u,h^j(u))$ be the edge between $j$ and $v=h^j(u)$ that corresponds to $u$ in $G_i$---since $f(u) \geq f(h^j(u))$, this edge exists. Observe that 
one possible fractional matching for $G_i$ is to assign a fraction of $1 / (k - i + 1)$ for every edge of the set
\[
	R = \{e(j, u, h^j(u)) \mid 1 \leq j \leq k \text{ and } u \in M^j_i\}
	\enspace.
\]
Since $M^j_i$ is a set of size $k - i + 1$, $R$ contains $k - i + 1$ edges hitting the right side vertex $j$ of $G$ for every $1 \leq j \leq k$. Moreover, since every element of $B$ appears in exactly $k - i + 1$ of the sets $B^1_{i - 1}, B^2_{i - 1}, \dotsc, B^k_{i - 1}$ by the induction hypothesis, the number of $R$ edges hitting every left side vertex $u \in B$ of $G_i$ is also $k - i + 1$. Thus, we get that $R$ is a perfect fractional matching of $G_i$. Since the matching polytope is integral, there must be an integral perfect matching $R_i$ in $G_i$ whose weight is at least the weight of $R$, \ie,
\[
	\sum_{e(j, u, v) \in R_i} f(u \mid A^j_{i-1})
	\geq
	\frac{1}{k - i + 1} \cdot \sum_{j = 1}^k \sum_{u \in M_i^j} f(u \mid A^j_{i-1})
	\enspace.
\]

To complete the proof of the lemma, it remains to observe two things. First, note that since $R_i$ is a perfect matching, exactly one edge of $R_i$ hits every left side vertex $v \in B$ of $G_i$, and thus the number of appearances of $v$ in $B^1_i, B^2_i, \dotsc, B^k_i$ is smaller than the number of its appearances in $B^1_{i - 1}, B^2_{i - 1}, \dotsc, B^k_{i - 1}$ by exactly $1$ (if the edge $e(j, u, v)$ is the single edge of $R_i$ hitting $v$, then $B^j_{i - 1}$ contains $v$, but $B^j_i$ does not). The final observation is that, for every $j = 1,\dotsc, k$, $A^j_i \cupdot B^j_i = (A^j_{i-1} + u_i^j) \cupdot (B^j_{i-1} - v_i^j)$ is a base of $\cM$ because $e = (v_i^j, j)$ is an edge of $G_i$.
\end{proof}

Next, we fix a base $T$ of $\cM$, and construct additional sets $T_i^j$ such that $A_i^j \cupdot T_i^j$ is a base of $\cM$ for every $0 \leq i \leq k$ and $1 \leq j \leq k$. 
For $i = 0$ we define $T_0^j = T$ for every $1 \leq j \leq k$. Consider now some $1 \leq i,j \leq k$ and let us construct $T_i^j$ assuming $T_{i-1}^j$ is already constructed. $M_i^j$ and $T_{i-1}^j$ are both bases of $\cM / A_{i-1}^j$, and $M_i^j$ maximizes the linear function $w(u) = f(u \mid A^j_{i-1})$. Thus, by Lemma~\ref{lem:matroid_matching_greedy}, there exists a bijection $h_i^j\colon M_i^j \to T_{i-1}^j$ such that $(T_{i-1}^j - h(u)) + u$ is a base of $\cM / A_{i-1}$ and $f(u \mid A_{i-1}^j) \geq f(h(u) \mid A_{i-1}^j)$ for every $u \in M_i^j$. Since $A_i^j = A_{i - 1}^{j}+u_i^j$, by setting $T_i^j = T_{i-1}^j - h_i^j(u)$ we get that $A_i^j \cupdot T_i^j$ is a base of $\cM$ as promised.


The following lemma is analogous to Lemma~\ref{lem:A_OPT_development}, and it implies Corollary~\ref{cor:main_inequality_deterministic}, which is analogous to Corollary~\ref{cor:mainineq}.
\begin{lemma} \label{lem:A_OPT_development_deterministic}
For every $0 \leq i \leq k$, $1 \leq j \leq k$ and set $S \subseteq \cN$,
\[
	f(A_i^j) + f(A_i^j \cup T_i^j \cup S)
	\geq
	f(A_{i - 1}^j) + f(A_{i - 1}^j \cup T_{i-1}^j \cup S)
\]
and
\[
	f(A_i^j) + f(A_i^j \cup B_i^j \cup S)
	\geq
	f(A_{i - 1}^j) + f(A_{i - 1}^j \cup B_{i-1}^j \cup S)
	\enspace.
\]
\end{lemma}
\begin{proof}
Observe that
\begin{align*}
	f(A_i^j) - f(A_{i-1}^j)
	&=
	f(u_i^j \mid A_{i-1}^j)
	\geq
	f(h_i^j(u^i_j) \mid A_{i-1}^j)
	\geq
	f(h_i^j(u^i_j) \mid A_{i-1}^j \cup T_i^j \cup S)\\
	={} &
	f(A_{i - 1}^j \cup T_{i - 1}^j \cup S) - f(A_{i-1}^j \cup T_i^j \cup S)
	\geq
	f(A_{i - 1}^j \cup T_{i - 1}^j \cup S) - f(A_i^j \cup T_i^j \cup S)
	\enspace,
\end{align*}
where the first inequality holds due to the definition of $h_i^j$, the second follows from the monotonicity and submoduarlity of $f$ and the last from $f$'s monotonicity. Similarly,
\begin{align*}
	f(A_i^j) - f(A_{i-1}^j)
	&=
	f(u_i^j \mid A_{i-1}^j)
	\geq
	f(v_i^j \mid A_{i-1}^j)
	\geq
	f(v_i^j \mid A_{i-1}^j \cup B_i^j \cup S)\\
	={} &
	f(A_{i - 1}^j \cup B_{i - 1}^j \cup S) - f(A_{i-1}^j \cup B_i^j \cup S)
	\geq
	f(A_{i - 1}^j \cup B_{i - 1}^j \cup S) - f(A_i^j \cup B_i^j \cup S)
	\enspace,
\end{align*}
where the first inequality holds since for every edge $e = (v, j)$ of $G_i$ the element $u \in M_i^j$ corresponding to this edge obeys $f(u \mid A_{i-1}^j) \geq f(v \mid A_{i-1}^j)$.
\end{proof}
\begin{corollary} \label{cor:main_inequality_deterministic}
For every $0 \leq i \leq k$, $1 \leq j \leq k$ and set $S \subseteq \cN$,
\begin{equation} \label{eq:main_inequality_T}
	f(A^j_k) + f(A^j_k \cup S)
	\geq
	f(A_i^j) + f(A_i^j \cup T_i^j \cup S)
	\geq
	f(T \cup S)
\end{equation}
and
\begin{equation} \label{eq:main_inequality_B}
	f(A^j_k) + f(A^j_k \cup S)
	\geq
	f(A_i^j) + f(A_i^j \cup B_i^j \cup S)
	\geq
	f(B \cup S)
	\enspace.
\end{equation}
\end{corollary}
\begin{proof}
The first part of Lemma~\ref{lem:A_OPT_development_deterministic} shows that $f(A^j_i) + f(A^j_i \cup T_i^j \cup S)$ is a non-decreasing function of $i$, and thus
\[
	f(A^j_k) + f(A^j_k \cup T_k^j \cup S)
	\geq
	f(A^j_i) + f(A^j_i \cup T_i^j \cup S)
	\geq
	f(A^j_0) + f(A^j_0 \cup T_0^j \cup S)
	\enspace.
\]
The first part of the corollary now follows by recalling that $A^j_0 = \varnothing$ and $T_0^j = T$ by definition, observing that $f(A^j_0) \geq 0$ since $f$ is non-negative and oberving that $T^j_k = \varnothing$ since $A^k_j$ and $A^k_j \cupdot T^j_k$ are both bases of $\cM$.

The second part of the corollary follows from the second part of Lemma~\ref{lem:A_OPT_development_deterministic} in a similar way.
\end{proof}

By setting $S = \varnothing$, the last corollary gives us a lower bound of $f(T)/2$ on $k^{-1} \cdot \sum_{j = 1}^k f(A_k^j)$. An alternative lower bound on this quantity is given by the next lemma, which is analogous to Lemma~\ref{lem:basic_analysis}.
\begin{lemma} \label{lem:basic_analysis_deterministic}
For every $0 \leq i \leq k$, $k^{-1} \cdot \sum_{j = 1}^k f(A_i^j) \geq [g(\nicefrac{i}{k}) + \delta] \cdot f(T)$, where $\delta = 1/(2k^2)$ for $0 < i < k$ and $0$ otherwise. \end{lemma}
\begin{proof}
Since $g(1) = 1/2$, the above discussion implies that the lemma follows from Corollary~\ref{cor:mainineq} in the case $i = k$. We prove the lemma for the other cases by induction. For $i = 0$ the lemma holds due to the non-negativity of $f$ since $g(0) = 0$. The rest of the proof is devoted to showing that the lemma holds for $1 \leq i < k$ given that it holds for $i - 1$.

Observe that the expression $\sum_{j = 1}^j [f(A_i^j) - f(A_i^{j-1})]$ is equal to the weight of the matching $R_i$, which is at least $(k - i + 1)^{-1} \cdot \sum_{j = 1}^k \sum_{u \in M_i^j} f(u \mid A_{i-1}^j)$ by Lemma~\ref{lem:basic_deterministic_properties}. Thus,
\begin{align*}
	\sum_{j = 1}^k [f(A_i^j) - f(A_i^{j-1})]
	\geq{} &
	\frac{\sum_{j = 1}^k \sum_{u \in M_i^j} f(u \mid A_{i-1}^j)}{k - i + 1}
	\geq
	\frac{\sum_{j = 1}^k \sum_{u \in T_{i - 1}^j} f(u \mid A_{i-1}^j)}{k - i + 1}\\
	\geq{} &
	\frac{\sum_{j = 1}^k f(T_{i - 1}^j \mid A_{i-1}^j)}{k - i + 1}
	\geq
	\frac{\sum_{j = 1}^k [f(T) - 2f(A_{i-1}^j)]}{k - i + 1}
	\enspace,
\end{align*}
where the second inequality follows from the definition of $M_i^j$, the second inequality follows from the submodularity of $f$ and the last inequality follows from the right inequality of~\eqref{eq:main_inequality_T} by choosing $S = \varnothing$.

Rearranging this inequality, we get
{\mleftright
\begin{align*}
	\sum_{j = 1}^k f(A^j_i)
	\geq{} &
	\sum_{j = 1}^k f(A^j_{i-1}) + \frac{\sum_{j = 1}^k [f(T) - 2f(A_{i-1}^j)]}{k - i + 1}
	=
	\frac{k - i - 1}{k - i + 1} \cdot \sum_{j = 1}^k f(A^j_{i-1}) + \frac{k \cdot f(T)}{k - i + 1}\\
	\geq{} &
	\frac{k(k - i - 1)}{k - i + 1} \cdot g\left(\frac{i - 1}{k}\right) \cdot f(T) + \frac{k \cdot f(T)}{k - i + 1}
	=
	k\left[g\left(\frac{i - 1}{k}\right) + \frac{1 - 2g(\frac{i-1}{k})}{k - i + 1} \right] \cdot f(T)
	\enspace,
\end{align*}
}%
where the second inequality follows from the induction hypothesis (since $i < k$). Using the observation that the derivative $g'(x)$ of $g(x)$ obeys $1 - x = g'(x) = (1 - 2g(x)) / (1 - x)$, the last inequality yields
{\mleftright
\begin{align*}
	\frac{k^{-1} \cdot \sum_{j = 1}^k f(A^j_i)}{f(T)}
	\geq{} &
	g\left(\frac{i - 1}{k}\right) + \frac{1 - 2g(\frac{i-1}{k})}{k - i + 1}
	=
	g\left(\frac{i - 1}{k}\right) + \frac{g'(\frac{i-1}{k})}{k}\\
	={} &
	g\left(\frac{i - 1}{k}\right) + \int_{(i - 1)/k}^{i/k} g'(x) dx + \int_{(i - 1)/k}^{i/k} \left[g'\left(\frac{i-1}{k}\right) - g'(x)\right] dx\\
	={}&
	g(\nicefrac{i}{k}) + \int_{(i - 1)/k}^{i/k} \left(x - \frac{i-1}{k}\right) dx
	=
	g(\nicefrac{i}{k}) + \frac{1}{2k^2}
	\enspace.
	\qedhere
\end{align*}
}%
\end{proof}

We now need a method to lower bound the gain of Algorithm~\ref{alg:deterministic_greedy} in its last iterations. The next lemma proves such a lower bound, and Corollary~\ref{cor:combined_bound_deterministic} combines this bound with the previous lemma to get a powerful lower bound on the quality of the output of Algorithm~\ref{alg:deterministic_greedy}.
\begin{lemma} \label{lem:final_gain_deterministic}
For every $0 \leq i \leq k$,
\[
	k^{-1} \cdot \sum_{j = 1}^k f(A^j_k \mid A^j_i)
	\geq
	(1 - \nicefrac{i}{k}) \cdot f(B \mid T) - 2k^{-1} \cdot \sum_{j = 1}^k f(A^j_k) + f(T)
	\enspace.
\]
\end{lemma}
\begin{proof}
Observe that, for every $1 \leq j \leq k$,
\begin{align*}
	f(A^j_k \mid A^j_i)
	\geq{} &
	f(B_i^j \mid A^j_k \cup T)
	=
	f(B_i^j \cup A^j_k \mid T) - f(A^j_k \mid T)]
	\geq
	f(B_i^j \mid T) - f(A^j_k \mid T)\\
	={} &
	f(B_i^j \mid T) - f(A^j_k \cup T) + f(T)
	\geq
	f(B_i^j \mid T) - 2f(A^j_k) + f(T)
	\enspace,
\end{align*}
where the first inequality follows from the left side of~\eqref{eq:main_inequality_B} by plugging $S = A^j_k \cup T$, the second inequality follows from monotonicity and the last follows from Inequality~\eqref{eq:main_inequality_T} by setting $S = A^j_k$.

Let $B_i$ be a uniformly random set picked out of $B_i^1, B_i^2, \dotsc, B_i^k$. Recall that every element of $B$ appears in exactly $k - i$ out of these sets by Lemma~\ref{lem:basic_deterministic_properties}, and thus $B_i$ contains every such element with probability $1 - \nicefrac{i}{k}$. Together with Lemma~\ref{lem:sampling}, this yields
\[
	\bE[f(B_i \mid T)]
	\geq
	\nicefrac{i}{k} \cdot f(\varnothing \mid T) + (1 - \nicefrac{i}{k}) \cdot f(B \mid T)
	=
	(1 - \nicefrac{i}{k}) \cdot f(B \mid T)
	\enspace.
\]

Combining the two inequalities that we have proved, we get
\begin{align*}
	k^{-1} \cdot \sum_{j = 1}^k f(A^j_k \mid A^j_i)
	\geq{} &
	k^{-1} \cdot \sum_{j = 1}^k [f(B_i^j \mid T) - 2f(A^j_k) + f(T)]\\
	={} &
	\bE[f(B_i \mid T)] - 2k^{-1} \cdot \sum_{j = 1}^k f(A^j_k) + f(T)\\
	\geq{} &
	(1 - \nicefrac{i}{k}) \cdot f(B \mid T) - 2k^{-1} \cdot \sum_{j = 1}^k f(A^j_k) + f(T)
	\enspace.
	\qedhere
\end{align*}
\end{proof}
\begin{corollary} \label{cor:combined_bound_deterministic}
For every $0 \leq x \leq 1$, $3k^{-1} \cdot \sum_{j = 1}^k f(A^j_k) \geq (1 + g(x)) \cdot f(T) + (1 - x) \cdot f(B \mid T)$.
\end{corollary}
\begin{proof}
Let $i_1 = \lfloor xk \rfloor$ and $i_2 = \lceil xk \rceil$. Clearly, there must be a value $\alpha \in [0, 1]$ such that $x = [\alpha i_1 + (1 - \alpha) i_2]/k$. Additionally, since we assume that $k \geq 2$, at least one of the values $i_1$ or $i_2$ belongs to $(0, k)$. Thus, by Lemma~\ref{lem:basic_analysis_deterministic},
\begin{align*}
	\mspace{63mu}&\mspace{-63mu}\frac{\alpha k^{-1} \cdot \sum_{j = 1}^k f(A_{i_1}^j) + (1 - \alpha) k^{-1} \cdot \sum_{j = 1}^k f(A_{i_2}^j)}{f(T)}\\
	\geq{} &
	\alpha \cdot g(\nicefrac{i_1}{k}) + (1-\alpha) \cdot g(\nicefrac{i_2}{k}) + \min\{\alpha, 1 - \alpha\} / (2k^2)\\
	={} &
	\frac{\alpha i_1}{k} - \frac{\alpha i_1^2}{2k^2} + \frac{(1 - \alpha)i_2}{k} - \frac{(1 - \alpha) i_2^2}{2k^2} + \frac{\min\{\alpha, 1 - \alpha\}}{2k^2}\\
	={} &
	\frac{\alpha i_1 + (1 - \alpha)i_2}{k} - \frac{(\alpha i_1 + (1 - \alpha) i_2)^2}{2k^2} - \frac{\alpha(1 - \alpha)(i_1 - i_2)^2}{2k^2} + \frac{\min\{\alpha, 1 - \alpha\}}{2k^2}
	\enspace.
\end{align*}
Observe now that $(i_1 - i_2)^2 \in \{0, 1\}$ and $\min\{\alpha, 1 - \alpha\} \geq \alpha(1-\alpha)$. Thus, the last term on the rightmost expression in the last inequality is at least as large as the term before it, which gives us
\begin{align*}
	\alpha k^{-1} \cdot \sum_{j = 1}^k f(A_{i_1}^j) + (1 - \alpha) k^{-1} \cdot \sum_{j = 1}^k f(A_{i_2}^j)
	\geq{} &
	\left[\frac{\alpha i_1 + (1 - \alpha)i_2}{k} - \frac{(\alpha i_1 + (1 - \alpha) i_2)^2}{2k^2}\right] \cdot f(T)\\
	={} &
	(x - x^2) \cdot f(T)
	=
	g(x) \cdot f(T)
	\enspace.
\end{align*}

By Lemma~\ref{lem:final_gain_deterministic}, we also get
\begin{align*}
	\alpha k^{-1} \cdot \sum_{j = 1}^k f(A_k^j \mid A_{i_1}^j) + (1 - \alpha) k^{-1} &\cdot \sum_{j = 1}^k f(A_k^j \mid A_{i_2}^j)\\
	\geq{} &
	\alpha\mleft[(1 - \nicefrac{i_1}{k}) \cdot f(B \mid T) - 2k^{-1} \cdot \sum_{j = 1}^k f(A^j_k) + f(T)\mright]\\
	&+
	(1 - \alpha)\mleft[(1 - \nicefrac{i_2}{k}) \cdot f(B \mid T) - 2k^{-1} \cdot \sum_{j = 1}^k f(A^j_k) + f(T)\mright]\\
	={} &
	 (1 - x) \cdot f(B \mid T) - 2k^{-1} \cdot \sum_{j = 1}^k f(A^j_k) + f(T)
	\enspace.
\end{align*}
The corollary now follows by adding this inequality to the previous one and rearranging.
\end{proof}

The following proposition summarizes some of the properties we have proved for Algorithm~\ref{alg:deterministic_greedy}.
\begin{proposition} \label{prop:deterministic_properties}
For every base $T$ of $\cM$, the output set $A$ of Algorithm~\ref{alg:deterministic_greedy} obeys
\begin{compactenum}
	\item $f(A) \geq f(T)/2$.
	\item $f(A) \geq \frac{(1 + g(x)) \cdot f(T) + (1 - x) \cdot f(B \mid T)}{3}$ for every $x \in [0, 1]$.
\end{compactenum}
\end{proposition}

\begin{proof}
The proposition follows by combining Lemma~\ref{lem:basic_analysis_deterministic} (for $i = k$) and Corollary~\ref{cor:combined_bound_deterministic} with the observation that the value of $A$ is at least the average value of $A_k^1, A_k^2, \dotsc, A_k^k$ since Algorithm~\ref{alg:deterministic_greedy} picks the best among these sets as $A$.
\end{proof}

We are now ready to prove Theorem~\ref{thm:main}.
\begin{proof}[Proof of Theorem~\ref{thm:main}]
Consider a modified version of Algorithm~\ref{alg:divide} in which the invocations of Algorithm~\ref{alg:ResidualRandomGreedy} have been replaced with invocations of Algorithm~\ref{alg:deterministic_greedy} (with the parameter $B$ set to $B_1$ in the first invocation and to $A_1$ in the other). 
The (fully deterministic) algorithm thus obtained is given as Algorithm~\ref{alg:divide_deterministic}.
\begin{algorithm}
\caption{\textsf{Matroid Split and Grow - Deterministic}$(f,\cM)$} \label{alg:divide_deterministic}
$(A_1,B_1) \gets \mbox{\Split}(f, \cM, p)$.\\
$A_2 \gets \text{\RPGreedy}(f(\cdot \mid A_1), \cM / A_1, B_1)$. \\
$B_2 \gets \text{\RPGreedy}(f(\cdot \mid B_1), \cM / B_1, A_1)$. \\
Return the better solution out of $A=(A_1 \cupdot A_2)$ and $B=(B_1 \cupdot B_2)$.
\end{algorithm}

Observe that Proposition~\ref{prop:deterministic_properties} is identical to Proposition~\ref{prop:random_greedy_properties} except for two differences. First, the inequalities now hold deterministically rather than in expectation. Second, the sets $T_1$ and $T_2$ have been replaced with $T$ and $B$, respectively. Thus, the the analysis of Algorithm~\ref{alg:divide} can be applied to Algorithm~\ref{alg:divide_deterministic} by simply replacing every use of Proposition~\ref{prop:random_greedy_properties} with a use of Proposition~\ref{prop:deterministic_properties}, which shows that Algorithm~\ref{alg:divide_deterministic} is also a $0.5008$-approximation algorithm. 
\end{proof}

%% file: MatroidTheorem.tex
\section{Proof of Lemma~\ref{lem:matroid_matching_greedy}}\label{sec:proof}

In this section, we prove Lemma~\ref{lem:matroid_matching_greedy}. We repeat the lemma here for convenience.

\begin{replemma}{lem:matroid_matching_greedy}
Let $A$ and $B$ be two bases of a matroid $\cM = (\cN, \cI)$, where $A$ is a maximum weight base according to some weight function $w\colon \cN \to \nnR$.
Then, there exist a bijective function $h\colon A \to B$ such that for every element $u \in A$
\begin{enumerate}
	\item $(B - h(u)) + u$ is a base of $\cM$.
	\item $w(u) \geq w(h(u))$.
\end{enumerate}
\end{replemma}

We prove Lemma~\ref{lem:matroid_matching_greedy} by induction on the rank of $\cM$, which is denoted by $k$ as usual. For $k=0$ the lemma is trivial. In the rest of the section we assume Lemma~\ref{lem:matroid_matching_greedy} holds for every matroid of rank $k - 1$ and prove it for $\cM$ (whose rank is $k$).
We need the following well known fact about matroids, which was proved by~\cite{B69} and can be found (with a different phrasing) as Theorem~39.12 in~\cite{S03}.
\begin{lemma} \label{lem:two_sided_exchange}
If $X$ and $Y$ are two bases of a matroid $\cM = (\cN, \cI)$, then for every element $x \in X \setminus Y$ there exists an element $y \in Y \setminus X$ such that both $(X - x) + y$ and $(Y - y) + x$ are bases.
\end{lemma}

Let $u_A$ be an arbitrary minimum weight element of the base $A$ according to the weight function $w$. Using the last lemma we get the following observation.
\begin{observation} \label{obs:exchange_pair}
There exists an element $u_B \in B$ such that
\begin{compactitem}
	\item $(A - u_A) + u_B$ and $(B - u_B) + u_A$ are both bases of $\cM$.
	\item $w(u_A) \geq w(u_B)$.
\end{compactitem}
\end{observation}
\begin{proof}
If $u_A$ appears in $B$, then the observation is satisfied by choosing $u_B = u_A$. Otherwise, by Lemma~\ref{lem:two_sided_exchange}, there must be an element $u_B \in B \setminus A$ which obeys the first part of the observation. 
Moreover, $w(u_A) \geq w(u_B)$ because otherwise the base $(A - u_A) + u_B$ has a higher value than $A$, contradicting the fact that $A$ is a maximum weight base according to the weight function $w$.
\end{proof}

Consider now the matroid $\cM / u_B$.
\begin{lemma}
$A - u_A$ is a maximum weight base of $\cM / u_B$ according to the weight function $w$.
\end{lemma}
\begin{proof}
Since $A$ is a maximum weight base of $\cM$ and $u_A$ is a minimum weight element of it, for an appropriately chosen tie breaking rule the greedy algorithm has the following two properties: it outputs $A$, and $u_A$ is the last element added by the greedy algorithm to its solution. One can verify that when the greedy algorithm is applied to the matroid $\cM / u_B$ with the same tie breaking rule, it will output the set $A - u_A$ (because it will make the same decisions when executed on either $\cM$ or $\cM / u_B$ up to the point where its solution is equal to the base $A - u_A$ of $\cM / u_B$). This implies the lemma since the greedy algorithm always outputs a maximum weight base.
\end{proof}

Combining the last lemma with the induction hypothesis, we get that there must exist a bijective function $h'\colon (A - u_A) \to (B - u_B)$ such that, for every $u \in A - u_A$, $(B - u_B - h'(u)) + u$ is a base of $\cM / u_B$  and $w(u) \geq w(h'(u))$. 
We can now construct the function $h\colon A \to B$ as follows. For every element $u \in A$,
\[
	h(u)
	=
	\begin{cases}
		u_B & \text{if $u = u_A$} \enspace,\\
		h'(u) & \text{otherwise} \enspace.
	\end{cases}
\]

The guarantee of $h'$ and the fact that $w(u_A) \geq w(u_B)$ by Observation~\ref{obs:exchange_pair} imply together that $w(u) \geq w(h(u))$ for every $u \in A$ as promised. Also, $h$ is clearly a bijective function since $h'$ is a bijective function and $h(u_A) = u_B$  is not in the range of $h'$. Finally, we observe the following.

%

\begin{lemma}
For every $u \in A$, $(B - h(u)) + u$ is a base of $\cM$.
\end{lemma}
\begin{proof}
For $u = u_A$, the lemma follows from Observation~\ref{obs:exchange_pair} since $h(u_A) = u_B$. Otherwise, we have $u \in A - u_A$, which implies, by the guarantee of $h'$, that $(B - u_B - h(u)) + u = (B - u_B - h'(u)) + u$ is a base of $\cM / u_B$. Thus, by the properties of contraction,
\[
	[(B - u_B - h(u)) + u] + u_B
	=
	(B - h(u)) + u
\]
is a base of $\cM$ (the equality holds since
the fact that $h$ is bijective implies $h(u) \neq h(u_A) = u_B$).
\end{proof}